\documentclass[pdflatex,sn-mathphys-num]{sn-jnl}

\usepackage{graphicx}%
\usepackage{multirow}%
\usepackage{amsmath,amssymb,amsfonts}%
\usepackage{amsthm}%
\usepackage{mathrsfs}%
\usepackage[title]{appendix}%
\usepackage[dvipsnames]{xcolor}%
\usepackage{textcomp}%
\usepackage{manyfoot}%
\usepackage{booktabs}%
\usepackage{algorithm}%
\usepackage{algorithmicx}%
\usepackage{algpseudocodex}%
\usepackage{listings}%

\usepackage[nice]{nicefrac}
\usepackage{standalone} 
\usepackage{enumerate}
\usepackage{dsfont}
\usepackage{tikz}
\usepackage{fge}
\usepackage[switch]{lineno}
\usepackage{pifont}
\usepackage{natbib}
\usepackage{hyperref}
\usepackage{caption}
\usepackage{subcaption}
\usepackage{lipsum}

\definecolor{TODO-blue}{HTML}{0740f8}
\definecolor{Rescuer-Orange}{HTML}{fb5604}

\newcommand{\NP}{\textbf{NP}}
\newcommand{\RP}{\textbf{RP}}
\newcommand{\cmark}{\ding{51}\;}
\newcommand{\xmark}{\ding{55}\;}

\newcommand{\MM}{\mathbf{M}}
\newcommand{\PP}{\mathbb{P}}

\renewcommand{\ss}{\mathbf{s}}
\newcommand{\rr}{\mathbf{r}}

\newcommand{\cC}{\mathcal{C}}
\newcommand{\dD}{\mathcal{D}}

\newcommand{\hH}{\mathcal{H}}
\newcommand{\lL}{\mathcal{L}}

\newcommand{\oO}{\mathcal{O}}

\newcommand{\argmax}{\mathop{\mathrm{argmax}}}

\newcommand{\ceil}[1]{\lceil#1\rceil}
\newcommand{\Ind}[1]{\mathds{1}(#1)} 
\newcommand{\mysetminus}{\raisebox{0.05em}{\scriptsize $\hspace{0.02cm} \fgebackslash \hspace{0.02cm}$}} 
\newcommand{\AW}{\alpha(W)} 

\newtheorem{cor}{Corollary} 

\newcommand{\majority}[1]{\tikz[baseline]{\draw[fill=#1,line width=0.5pt] (0,0.1) circle(0.9ex)}}

\newcommand{\unfeasible}{\tikz[baseline]{\draw[color=BrickRed, line width=0.8pt, fill=white] (0,-0.04)--(0,0.26)--(0.3,0.26)--(0.3,-0.04)--(0,-0.04);}\,}
\newcommand{\rescue}{\tikz[baseline]{\draw[color=Rescuer-Orange, line width=0.5pt, fill=Rescuer-Orange!45] (0,-0.04)--(0,0.26)--(0.3,0.26)--(0.3,-0.04)--(0,-0.04);}\,}
\newcommand{\rescuer}{\tikz[baseline]{\draw[color=Rescuer-Orange, line width=0.5pt, fill=Rescuer-Orange] (0,0.105) circle(1.1ex);
\draw [color=white,draw opacity=1 ][line width=3]    (-0.15,0.25) -- (0.15,-0.05) ;
\draw [color=white,draw opacity=1 ][line width=3]    (-0.15,-0.05) -- (0.15,0.25) ;
\draw[color=Rescuer-Orange, line width=0.7pt, fill=white] (0,0.105) circle(0.5ex);
\draw[color=Rescuer-Orange, line width=0.6pt] (0,0.105) circle(1.1ex);}\;}

\theoremstyle{thmstyleone}%
\newtheorem{theorem}{Theorem}
\newtheorem{proposition}[theorem]{Proposition}%

\theoremstyle{thmstyletwo}%
\newtheorem{example}{Example}%
\newtheorem{remark}{Remark}%

\theoremstyle{thmstylethree}%
\newtheorem{definition}{Definition}%

\begin{document}

\title[On the Limits of PAC Learning of Networks from Opinion Dynamics]{On the Limits of PAC Learning of Networks from Opinion Dynamics}


\author[1,2]{\fnm{Dmitry} \sur{Chistikov}}\email{d.chistikov@warwick.ac.uk}
\equalcont{These authors contributed equally to this work.}

\author*[3]{\fnm{Luisa} \sur{Estrada}}\email{luisa-fernanda.estrada-plata@warwick.ac.uk}

\author[1,2]{\fnm{Mike} \sur{Paterson}}\email{m.s.paterson@warwick.ac.uk}
\equalcont{These authors contributed equally to this work.}

\author[2]{\fnm{Paolo} \sur{Turrini}}\email{p.turrini@warwick.ac.uk}
\equalcont{These authors contributed equally to this work.}

\affil[1]{\orgdiv{Centre for Discrete Mathematics and its Applications}, \orgname{University of Warwick}, \orgaddress{\country{United Kingdom}}}

\affil[2]{\orgdiv{Department of Computer Science}, \orgname{University of Warwick}, \orgaddress{\country{United Kingdom}}}

\affil[3]{\orgdiv{Department of Mathematics}, \orgname{University of Warwick}, \orgaddress{\country{United Kingdom}}}


\abstract{Agents in social networks with threshold-based dynamics change opinions when influenced by sufficiently many peers. Existing literature typically assumes that the network structure and dynamics are fully known, which is often unrealistic. In this work, we ask how to learn a network structure from samples of the agents' synchronous opinion updates. Firstly, if the opinion dynamics follow a threshold rule in which a fixed number of influencers prevent opinion change (e.g., \emph{unanimity} and \emph{quasi-unanimity}), we provide an efficient PAC learning algorithm provided that the number of influencers per agent is bounded. Secondly, under standard computational complexity assumptions, we prove that if agents' opinions follow the \emph{majority} of their influencers, then there is no efficient PAC learning algorithm. We propose a polynomial-time heuristic that successfully learns consistent networks in over $98\%$ of our simulations on random graphs, with no failures for some specified conditions on the numbers of agents and opinion diffusion examples.}

\keywords{Social Networks, PAC learning, Graph Inference, Opinion Diffusion, Threshold Dynamics}



\maketitle

\section{Introduction}

Opinions are ubiquitous, and learning how they emerge and evolve is crucial for understanding the spread of ideas, beliefs and behaviours within societies. From the rise of social movements to the diffusion of technological innovations, the mechanisms driving opinion change have profound implications in economics \cite{jackson_evolution_2002}, sociology \cite{kandler_modeling_2017}, and political science \cite{castiglioni_election_2019}. Opinions are also contagious, making them practical tools for developing and analysing models in epidemiology \cite{leskovec_cost-effective_2007} and marketing \cite{deligkas_being_2023}. They are shaped and reshaped through interactions among agents in social networks, contributing to processes of cultural transmission \cite{fogarty_role_2013, galesic_beyond_2023} and social learning \cite{ammar_social_2023, gavrilets_solution_2014}.
    
From a computational perspective, opinion dynamics represent information propagating through a directed graph. This offers a framework for studying how social influence spreads in a network \cite{chistikov_convergence_2020,mukhopadhyay_voter_2020,auletta_minority_2015} and how it can be exploited for campaigning purposes \cite{kempe_maximizing_2003,bredereck_manipulating_2017,kushwaha_capricious_2022}.

Nevertheless, opinion dynamics studies frequently assume a known network structure. This is often unrealistic \cite{netrapalli_learning_2012}, especially when the structure is hidden for data protection. When this is the case, connections among agents remain unidentifiable as multiple networks can have identical dynamics over certain opinion inputs. Some existing studies may not assume the structure to be fully available but rely on the network being derived from a Stochastic Block Model \cite{wilder_maximizing_2018}. Others use Bayesian optimisation techniques to predict the opinions' transition matrix by empirically estimating the cross-correlation matrix \cite{ravazzi_learning_2021}. A polynomial-time algorithm was recently given for exact learning social networks endowed with majority dynamics, where opinions can be observed and manipulated \cite{chistikov_learning_2024}. While this last framework is close in spirit to ours, it requires learning a social network with certainty and relies on direct intervention on the agents' opinions. 

\subsection{Our Contribution.}

We study the problem of probably approximately correct (PAC) learning the structure of a social network from its threshold-based opinion dynamics. We ask whether there are algorithms that can learn a network with error at most $\varepsilon$ and confidence at least $1-\delta$. When opinion dynamics follow the \emph{all-but-$\kappa$} threshold dynamics, i.e., if agreement from a fixed number $\kappa$ of influencers inhibits opinion change, we give an efficient learning algorithm by constructing a Consistent Hypothesis Finder (Proposition \ref{prop: CHF unanimity} and Theorem \ref{theo: all but k solution}).
When the agents' opinions follow the \emph{majority} of their influencers, we prove under standard computational complexity assumptions ($\NP \ne \RP$)
that no polynomial-time algorithm can PAC learn the underlying network (Theorem \ref{theo: Hitting set reduction}). To complement this hardness result, we design a heuristic (Algorithm \ref{alg: Waterfall algorithm}) and evaluate its performance on synthetic random graphs. 

\subsection{Related Literature.}

Our approach connects two prominent research lines: opinion dynamics \cite{auletta_complexity_2020,bredereck_manipulating_2017,kempe_maximizing_2003,chistikov_convergence_2020} and network inference \cite{netrapalli_learning_2012,myers_convexity_2010,gomez-rodriguez_uncovering_2011}. For a more comprehensive survey, see \cite{grabisch_survey_2020}. 

We narrow our scope to models with binary opinions governed by threshold-based diffusion rules, as these embody fundamental human decision-making factors such as homophily and biased assimilation \cite{ravazzi_learning_2021}. In these models, agents form their opinions after tallying those of others in a weighted sum. They may have a preferred opinion \cite{mukhopadhyay_voter_2020} or attempt to conform with nearby agents, leading to Ising models inspired by the alignment of atom spins \cite{baldassarri_ising_2023}. Agents might adopt the most popular opinion from a random sample of neighbours, as in Voting Models \cite{liggett_interacting_2004}, or external forces can be used to add a manipulation layer \cite{bredereck_manipulating_2017}. A recent contribution \cite{chistikov_learning_2024} explored the exact learning of social networks through direct intervention on opinions, but relied on having polynomially many interventions available. We remove these requirements by studying random samples of individual opinion diffusion steps and relaxing the exact learning condition to an approximately correct one.

Possibly the most prominent example of threshold-based dynamics is the Linear Threshold model with active or inactive agents \cite{kempe_influential_2005}. In it, agents need to attain a fixed number of active neighbours to activate, and cannot be deactivated. This model has been pivotal in the area of influence maximisation on networks and has inspired several efficient algorithms, such as DD \cite{chen_efficient_2009}, CELF \cite{leskovec_cost-effective_2007}, TIM, and TIM+ \cite{tang_influence_2015}. 

We are interested in the boundary between tractable and intractable problems for these models. This was previously explored in \cite{auletta_complexity_2020,deligkas_being_2023} in the context of picking an optimal set of agents to spread an opinion. The authors showed that the problem's $\NP$-hardness depends on social network features, such as the presence of tree-like structures or cycles with sufficiently small periods. Similarly, 
\cite{narasimhan_learnability_2015} showed that PAC learning influence functions can be achieved with polynomial sample complexity when the activation times of the nodes are known; however, without this knowledge, the problem becomes computationally hard. In contrast, \cite{gomez-rodriguez_inferring_2010} showed, in an epidemiological context, that identifying the most likely infection structure is $\NP$-hard when contagion times are known but the diffusion rule is not. Authors in \cite{qiu_learning_2024} extended this line of work by studying the simultaneous PAC learning of both network topology and interaction functions, showing that the problem is generally computationally intractable; yet, the special case of threshold dynamics and graphs containing a perfect matching is an efficiently learnable class. We aim to refine the boundary of PAC learnability for inferring connections between agents. We explore the conditions under which the problem is hard, trading off additional information from the interaction function against extending the class of learnable graphs.

\subsection{Paper Structure}

Section \ref{sec: math preliminaries} introduces the notation and preliminary concepts used throughout the paper. It defines the \emph{all-but-$\kappa$} and \emph{$\tau$-margin} opinion diffusion protocols, from which unanimity and majority dynamics arise as special cases. The section reviews the PAC learning framework and formalises the Social Network inference problem as a PAC learning task. It also presents the \emph{matching transformation}, our main tool for handling binary opinions in later proofs and algorithms. Section \ref{sec: Unanimity and All-but-k} presents our theoretical results for \emph{all-but-$\kappa$} dynamics, while Section \ref{sec: hitting set problem} establishes a hardness result for social networks equipped with majority dynamics. In Section \ref{sec: Waterfall algorithm}, we propose a heuristic, the \texttt{Waterfall} algorithm, to tackle majority dynamics and analyse its theoretical guarantees and empirical performance. We discuss our results and future directions in Section \ref{sec: discussion}. 

For readability, full proofs are deferred to Appendix \ref{sec: full proofs}. An alternative characterisation of feasible solutions under majority dynamics is presented in Appendix \ref{sec: deep Waterfall}, together with a tie-breaking subroutine for the \texttt{Waterfall} algorithm and illustrative examples. Technical specifications of the parameters used in our random graph generation models, as well as additional experiments across different sparsity regimes, are presented in Appendix \ref{sec: param specs}.

\section{Mathematical Preliminaries}
\label{sec: math preliminaries}
\subsection{Social Networks.} 
When agents update their opinions synchronously, the network can be decoupled. Hence, inferring a network can be parallelised into finding each agent's influencers. 
Throughout this paper, let $N=[n]$ be a finite set of agents and add agent $i$ as a distinguished target node. Let $G$ be a directed graph over $N\cup \{i\}$ whose edges represent influence relations. 
An edge $j\to i$ in $G$ indicates that agent $j\in N$ \emph{influences} agent $i$. We refer to a function $\ell: N \to L$ that assigns labels from a set $L$ to the agents in $N$ as a \emph{network labelling} (or simply a \emph{labelling}). 
We consider the binary case with $L:=\{\phi,\neg \phi\}$ and, slightly abusing notation, write $\ell(j)=\phi$ when opinions of agent $j\in N$ and agent $i$ agree, and $\ell(j)=\neg \phi$ when they disagree.

In short, a \emph{social network endowed with opinion dynamics} is a directed graph of agents and rules of how they update opinions. We denote them with the triple $(N\cup\{i\}, G, f)$, where $f$ is a diffusion protocol with respect to agent $i$. The diffusion protocols assert whether agent $i$ changes opinion after interacting with a labelling $\ell\in L^N$. We define protocols with boolean sentences that return $\mathit{True}$ if agent $i$ changes opinion and $\mathit{False}$ if not.

We focus on two synchronous threshold-based protocols for binary opinions: all-but-$\kappa$, when an agent adopts an opinion unless $\kappa$ or more influencers support the current one; and $\tau$-margin, when an agent adopts an opinion if disagreeing influencers exceed those agreeing by more than $\tau$. These are defined below for an abstract set of agents $F\subseteq N$. Later on, we intend to retrieve the true \emph{influencer set} $G_i$ for each agent $i$ from observing their threshold-based dynamics. 

\begin{enumerate}
    \item \textbf{All-but-$\kappa$:} Let $\kappa \geq 0$. An agent changes opinion 
    unless $\kappa$ or more of its influencers agree with it. For any network labelling $\ell\in L^N$ and $F\subseteq N$, the all-but-$\kappa$ protocol $f^{\leq \kappa}$ is given by the boolean sentence
    \begin{align*}
        f^{\leq \kappa}(\ell, F):= \sum_{j\in F} \Ind{\ell(j)=\phi} \leq \kappa.
    \end{align*}
    \emph{Unanimity dynamics} arise for $\kappa=0$, where an agent changes opinion only if all its influencers disagree with~it.
    
    \item \textbf{$\tau$-margin:} An agent changes opinion if the number of influencers who disagree with it exceeds those who agree by a margin greater than $\tau\geq 0$. For any network labelling $\ell\in L^N$ and $F\subseteq N$, the $\tau$-margin protocol $f^{+\tau}$ is given by the boolean sentence
    \begin{align*}
        f^{+\tau}(\ell, F):= \sum_{j\in F} \Ind{\ell(j)\neq \phi} - \Ind{\ell(j)=\phi} > \tau.
    \end{align*} 
    Choosing $\tau=0$ yields \emph{majority dynamics}, where an agent changes opinion when a strict majority of its influencers disagree with it. We use $f^+$ to denote this.
\end{enumerate}

Given a sample of network labellings, the main challenge is to deal with agents who disagreed (or agreed) with our target agent in some labellings where its opinion changed and in others where it did not. This begs the question of whether or not they caused the change in opinion. To address this, we divide the sample into two subsets, called the \emph{always-changing} and the \emph{never-changing} labellings. We aim to show how they can co-exist without contradicting each other.

\begin{definition}
    \label{def: always-never changing}
    Let $f$ be any diffusion protocol. We say a set of network labellings $\{\ell_k\}_{k=1}^m$ is \emph{always-changing} for $F\subseteq N$ if the sentence $f(\ell_k,F)$ is $\mathit{True}$ for $k=1, \dots, m$. It is \emph{never-changing} if all $m$ sentences are $\mathit{False}$.
\end{definition}

\subsection{PAC Learning.}

We formalise the \emph{Probably Approximately Correct} (PAC) learning following \cite[Ch. 1.2]{kearns_introduction_1994}. The goal is to design a learning model whose predictions are both highly probable and close to the true answer. Let $X$ be the \emph{domain} of elements encoding the learner's world. A \emph{concept} $c$ over $X$ is a boolean mapping $c: X \to \{0,1\}$ indicating whether $x\in X$ has a desired property. A \emph{concept class} $\cC$ is a collection of concepts within a \emph{hypothesis space} $\hH$. A \emph{learner} $\lL$ is an algorithm whose objective is to distinguish between positive and negative examples of a target concept $c$ chosen arbitrarily from $\cC$. Examples are given by an \emph{oracle} $\text{Ex}(c,\dD)$, assumed to draw an example $x$ from $X$ with distribution $\dD$ and return $c(x)$ in constant time. Note that $\lL$ cannot query the oracle directly.

\begin{definition}
    A concept class $\cC$ over a domain $X$ is \emph{$(\varepsilon, \delta)$-PAC learnable} using $\hH \supseteq \cC$ if a learner $\lL$ exists such that, given any inputs $\varepsilon, \delta \in (0, \frac{1}{2})$ and access to $\emph{Ex}(c,\dD)$, $\lL$ outputs a hypothesis concept $h\in \hH$ satisfying with probability $1-\delta$ that $\PP[h(x)\neq c(x)]\leq \varepsilon$, for every concept $c\in \cC$ and $x\in X$ randomly sampled with distribution $\dD$. The error probability $\PP$ is taken over the random examples drawn from $\emph{Ex}(c,\dD)$ and any internal randomisation of $\lL$. If these conditions are met, we also say $\lL$ \emph{$(\varepsilon,\delta)$-PAC learns} $\cC$.
\end{definition}

\begin{definition}
    A \emph{Consistent Hypothesis Finder} (CHF) for a concept class $\cC$ over $X$ using $\hH \supseteq \cC$ is an algorithm $\lL$ such that, for all $m>0$ and $c\in \cC$, if $\lL$ is given a sample $\{(x_1, c(x_1)), \dots, (x_m, c(x_m))\}$, 
    then it outputs a hypothesis $h\in \hH$ satisfying $h(x_k)=c(x_k)$ for $k=1, \dots, m$.
\end{definition}

The learner gains more information as the oracle $\text{Ex}(c, \dD)$ draws more examples from $X$ since fewer hypotheses in $\hH$ remain indistinguishable from the target concept $c\in \cC$. Yet, the learner cannot refine its guess to the correct concept if multiple hypotheses are consistent with the training examples. This is resolved by the Fundamental Theorem of PAC Learning (see, e.g., \cite[Cor. 3]{hanneke_optimal_2016}), which links the sample complexity of a concept class with its VC dimension. 

\begin{definition}
    For a concept class $\cC$ over $X$ and a finite set of points $S=\{x_1, \dots, x_m\}$, if $\{(c(x_1), \dots, c(x_m)): c\in \cC\}=\{0,1\}^m$, then we say that $S$ is \emph{shattered} by $\cC$. Furthermore, the cardinality of the largest set shattered by $\cC$ is its \emph{Vapnik-Chervonenkis} dimension $\text{VCD}(\cC)$.
    \label{def: VC dim}
\end{definition}

\begin{theorem}[Fundamental Theorem of PAC Learning]
    Let $\cC$ be an arbitrary concept class in $X$ and $\dD$ any distribution over $X$. Then, any algorithm that $(\varepsilon, \delta)$-PAC learns $\cC$ requires at least 
    \begin{align*}
      m=\Theta\left(\frac{1}{\varepsilon} \left( \text{VCD}(\cC) + \log \frac{1}{\delta} \right)\right)  
    \end{align*}
    many examples from $\text{Ex}(c, \dD)$, for any $c\in \cC$.
    \label{theo: tight bounds on PAC learning}
\end{theorem}

Theorem \ref{theo: tight bounds on PAC learning} provides an upper bound on the number of examples sufficient to attain a desired PAC learning error $\varepsilon$ and confidence $\delta$. Moreover, \cite{hanneke_optimal_2016} showed this bound is tight, as it matches earlier lower bounds from \cite{blumer_learnability_1989} and \cite{ehrenfeucht_general_1989}. Consequently, for a fixed sample size $m$, we can recover achievable values for $\varepsilon$ and $\delta$ up to multiplicative constants.

\subsection{The Social Network Inference Problem}

We consider a learner who receives a sample of opinion updates for a target node. We assume it knows the diffusion protocol but not the underlying network. The learner's guess for the network improves with more examples, but how many are sufficient for their guess to be good enough? 

Let us formulate the social network inference problem in terms of PAC learning. Define the domain $X$ as the space of network labellings $L^N$ with concept classes induced by a diffusion protocol, $f^{\leq \kappa}$ or $f^{+\tau}$. For a given agent $i$, the concept to learn is the influencer set $ G_i \subseteq N$ that drives its opinion updates. An oracle $\text{Ex}(G_i, \dD)$ draws network labellings $\{\ell_k\}_{k=1}^{m}$ according to a distribution $\dD$ over $L^N$ and returns each as a positive example if agent $i$ changes opinion (i.e., if $f(\ell_k, G_i)$ is $\mathit{True}$, for $f\in \{f^{\leq \kappa},\, f^{+\tau}\}$) or negative otherwise. Notice that both protocols are linear threshold functions over a domain with $n:=\vert N\vert$ dimensions. As such, they are well-studied concept classes with VC dimension $n+1$ (see, e.g., \cite[Ch. 10]{anthony_computational_1992}). Thus, building a CHF in this context boils down to finding a subset $F \subseteq N$ whose oracle prediction matches $G_i$ for the labellings sampled so far. Altogether, we can $(\varepsilon, \delta)-$PAC learn an agent's influencers if the subset $F \subseteq N$ matches the oracle prediction of $G_i$ for a labelling sample of size
\begin{equation}
  m=\Theta\left(\frac{1}{\varepsilon} \left( (n + 1) + \log \frac{1}{\delta} \right)\right).
  \label{eq: how many examples}
\end{equation}

\subsection{Matching Transformations.}

Working with opinions relative to a target agent allows us to efficiently format the inputs for our algorithms via what we call a \emph{matching transformation}. Intuitively, this transformation recovers the agents whose opinions ``match'' the oracle's prediction for agent $i$. If the oracle predicts agent $i$ changes opinion after a labelling $\ell \in L^N$, then the agents who disagreed with it (i.e., $j\in N$ s.t. $\ell(j)\neq \phi$) match this prediction. Conversely, if no change is predicted, then the agents who agreed with agent $i$ in $\ell$ (i.e., $j\in N$ s.t. $\ell(j)=\phi$) match the prediction.

\begin{definition}
    \label{def: matching set}
    Let $(N\cup\{i\}, G, f)$ be a social network with opinion diffusion protocol $f$. For a labelling $\ell\in L^N$, $L=\{\phi, \neg \phi\}$, the \emph{matching set} $M(\ell)\subseteq N$ is the subset of agents such that $j\in M(\ell)$ if $\ell (j) = \neg \phi$ when $f(\ell, G_i)$ is $\mathit{True}$, or $\ell(j) = \phi$ when $f(\ell, G_i)$ is $\mathit{False}$.
\end{definition}

\begin{definition}
    \label{def: matching transformation}
    For $m,n>0$, let $N=[n]$, $(N\cup\{i\}, G, f)$ be a social network, and $(\ell_k)_{k=1}^m$ be a sequence of network labellings. The \emph{matching transformation} for $(\ell_k)_{k=1}^m$ is given by the $m\times n$ matrix $\MM$ with entries 
    \begin{align*}
    \mathbf{M}_{k,j} = \begin{cases}
    \hspace{0.25cm}1 &\text{if} \quad j \in M(\ell_k) \quad \text{and}\\
    -1 & \text{otherwise,}
    \end{cases}
    \end{align*}
    where $M(\ell)\subseteq N$ is the matching set of $\ell\in (\ell_k)_{k=1}^m$.
\end{definition}

\begin{example}
    \label{ex: majority sets}
    Suppose there is a network under some opinion protocol $f$ with $N=[5]$ plus agent $i$. An oracle samples the following labellings $\ell_1, \ell_2 \in \{\phi, \neg \phi\}^N$ relative to agent $i$: 
    \begin{align}
        \begin{array}{c|rrrrr}
        N  & 1 & 2 & 3 & 4 & 5 \\
        \hline
        \ell_1 & \phi & \phi & \neg \phi & \neg \phi & \neg \phi \\
        \ell_2 & \neg \phi & \neg \phi & \phi & \phi & \neg \phi \\
        \end{array}
        \label{eq: ex labellings l1 l2}
    \end{align}
    The agents who disagree with agent $i$ are $\{3,4,5\}$ in $\ell_1$ and $\{1,2,5\}$ in $\ell_2$. The oracle knows which agents are the influencers $G_i\subseteq N$ and predicts agent $i$ changes opinion after $\ell_1$ but not after $\ell_2$. So, $f(\ell_1, G_i)$ is a positive example while $f(\ell_2, G_i)$ is a negative one. This yields the matching sets
    \begin{align*}
        &M(\ell_1) = \{j\in N : \ell(j)=\neg \phi\} \quad \text{and }\\
        &M(\ell_2) = \{j\in N : \ell(j)=\phi\}.
    \end{align*}
    We store these sets in the matching transformation
    \begin{align}
        \mathbf{M}= \begin{bmatrix}
        -1 & -1 & 1 & 1& 1\\
        -1 & -1 & 1 & 1& -1
        \end{bmatrix}.
        \label{eq: majority matrix example}
    \end{align}
    
Graphically, we use \majority{white} when an agent belongs to a matching set and \majority{black} when it does not. $\MM$ is shown in Figure \ref{fig: exm majority transformation}.
\begin{figure}[ht]
    \centering
    \includegraphics[width=0.45\columnwidth]{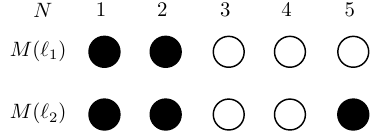}\\
    \caption{Matching transformation for $\MM$ as in \eqref{eq: majority matrix example}. }
    \label{fig: exm majority transformation}
    \end{figure}
\end{example}
The dimensions $m\times n$ of $\MM$ correspond to a labelling sample size $m$, and a social network with $n+1$ agents. Our learner's inputs are $\MM$ and the oracle's predictions $f(\ell_k, G_i)$ for $k=1, \dots, m$. We aim to find a set of agents $F\subseteq N$ whose opinion dynamics mimic the true set of influencers $G_i$ for our labelling sample. We refer to the agents that achieve this as \emph{feasible influencers}. 

\begin{remark}
    We write $\MM_F$ for the columns in $\MM$ restricted to agents in $F\subseteq N$. It is no coincidence that entries in 
    $\MM$ are $\pm 1$. This ensures that $F$ is a feasible influencer set under majority dynamics if and only if the row-sums in $\MM_F$ are nonnegative, and zero only when there is no opinion change.
\end{remark}

\section[Theoretical Results for All-But-k Dynamics]{Theoretical Results for All-But-$\kappa$ Dynamics}
\label{sec: Unanimity and All-but-k} 

Agents in networks ruled by an all-but-$\kappa$ protocol change their opinion unless more than $\kappa$ of their influencers already agree with them. Naturally, an agent changes opinion more often as $\kappa$ increases. We begin with the most restrictive case when $\kappa=0$. This yields \emph{unanimity~dynamics}, where agent $i$ changes opinion only when all its influencers disagree with it. In this case, the set of agents who matched the oracle's prediction across all the labellings where agent $i$ changed opinion will form a feasible influencer set. Thus, returning this set is a CHF with the following runtime upper bound.

\begin{proposition}
    Let $N=[n]$ and $(N\cup\{i\}, G, f^{\leq 0})$ be a social network under an all-but-$\kappa$ protocol, with $\kappa=0$. There exists a CHF that finds a feasible influencer set $F\subseteq N$ in $\oO(m\, n^2)$ time for any labelling sample of size~$m$. 
    \label{prop: CHF unanimity}
\end{proposition}
Next, we consider the all-but-$\kappa$ model with $\kappa > 0$, when all predictions state agent $i$ changes opinion. Note that selecting a subset with $\kappa+1$ or more agents in $M(\ell)^c$, $\ell \in L^N$, would be inconsistent for any labelling where the target agent changes opinion. So, if the learner checks all subsets in $N$ of size at most $\kappa+1$, it is bound to find at least one feasible influencer set, if such a set exists.

\begin{proposition}
    Let $N=[n]$ and $(N\cup\{i\}, G, f^{\leq \kappa})$ be a social network under an all-but-$\kappa$ protocol, $\kappa\geq 0$. There exists a CHF that finds a feasible influencer set $F\subseteq N$, if one exists, in $\oO(m\, n^{\kappa + 2})$ time for any always-changing labelling sample of size~$m$.
    \label{prop: always changing all but k}
\end{proposition}

The final step is to combine the always-changing and never-changing labellings. The learner knows that feasible sets must contain at least $\kappa$ influencers in the matching set of every never-changing labelling. So, if it checks all sets of $\kappa$ or more agents, progressively increasing in size, it will not need to check any set larger than the real influencer set $G_i$.

\begin{theorem}
    Let $N = [n]$ and $(N \cup \{i\}, G, f^{\leq \kappa})$ be a social network under an all-but-$\kappa$ protocol with $\kappa \geq 0$. If $G_i$ is the influencer set of agent $i$, then there exists a CHF that finds a feasible influencer set $F \subseteq N$ in  $\oO(m \, n^{|G_i| + 1})$ time for any labelling sample of size $m$. 
    \label{theo: all but k solution}
\end{theorem}
The complexity of the CHF depends on the size of the real influencer set $G_i$, not the $\kappa$ parameter in the all-but-$\kappa$ model. For networks where $\vert G_i\vert$ is bounded by a constant (e.g., sparse regular graphs), the complexity remains polynomial. However, the approach may become inefficient without this assumption. The upcoming section explores the boundary of efficient solutions for the majority dynamics case.\footnote{\; Notice that majority dynamics is a special case of the all-but-$\kappa$ and $\tau$-margin protocols, with either $\kappa = \ceil{\vert G_i\vert/2}-1$ or $\tau=0$.}

\section{Hardness for Learning Feasible Influencer Sets Under Majority Dynamics}
\label{sec: hitting set problem}

Our task for networks with majority dynamics is to find a set of agents $F\subseteq N$ such that at least half of them appear in every sampled labelling's matching set; strictly more if the target agent changes opinion, and half or more if not. A naive approach would return the intersection of all matching sets. But what if this intersection is empty? We apply the known computational hardness of the Hitting Set problem (see, e.g., \cite{dasgupta_algorithms_2006}) to prove that finding a feasible influencer set $F$ in these networks is $\NP$-complete. To do so, we adapt the approach in \cite{rudoy_np-hardness_2017}, and construct a matching transformation $\MM$ and an oracle prediction from a Hitting Set instance.

\begin{definition}[Hitting Set problem]
    Given a family of sets $\{S_1, S_2,..., S_{m}\}$ and a budget $d>0$, we wish to find, if possible, a set $C$ of size $d$ that has a non-empty intersection with every set $S_k$, $k=1,\dots, m$. \label{def: Hitting Set problem}
\end{definition}

\begin{theorem}
    The Hitting Set problem for $\{S_1, S_2,..., S_{m}\}$ with budget $d>0$ reduces to finding a feasible set of influencers for a target agent in a network with $n+1$ agents, where $n = \vert \bigcup_{k=1}^{m} S_k \vert + d + 1$, consistent with $m+d+2$ examples from a majority dynamics protocol. \label{theo: Hitting set reduction} 
\end{theorem}

\begin{proof}
    We encode any hitting set instance into a labelling sample for a network with majority dynamics. 
    Let $\{S_1, S_2,..., S_{m}\}$ be a family of sets and $d>0$ a budget. Define $\hat{n}:= \vert \bigcup_{k=1}^{m} S_k\vert$, $\hat{m}:=m+d+2$, $n:=\hat{n}+d+1$ and $N:=[n]$. 
    We construct a labelling sample $(\ell_k)_{k=1}^{\hat{m}}$ for a social network 
    $(N\cup \{ i\}, G, f^+)$ following \cite{rudoy_np-hardness_2017}. 
    
    We split the agents in $N$ into two: $\{a_1, \dots, a_{d+1}\}$ and $\{b_1, \dots, b_{\hat{n}}\}$. Further, we assign to each agent $b_j\in \{b_1, \dots, b_{\hat{n}}\}$ an element $\smash{s_j\in\bigcup_{k=1}^{m} S_k}$. The examples are built such that all auxiliary agents in $\{a_1, \dots, a_{d+1}\}$ are in the matching set of $\ell_{\hat{m}}$. For the other labellings, $\ell_1, \dots, \ell_{\hat{m}-1}$, agent $a_j$ belongs to $M(\ell_k)$ if 
    \begin{align*}
        \begin{cases}
         2\leq j\leq d+1 &\text{for } k=1, \dots, m, \text{ or }\\
         j=k-m &\text{for } k=m+1, \dots, \hat{m}-1.
        \end{cases}
    \end{align*}
    For the agents in $\{b_1, \dots, b_{\hat{n}}\}$, $b_j \in M(\ell_k)$ when $s_j\in S_k$ for $k=1, \dots, m$, plus for $k=m+1, \dots, \hat{m}-1$. See Figure \ref{fig: reduction as majority matrix} for a graphical example of a matching transform $\MM$ constructed in this way.

    We show that a hitting set $C\subseteq \bigcup_{k=1}^m S_k$, $\vert C\vert = d$, exists if and only if there is a feasible influencer set $F\subseteq N$ for which the target agent $i$ changes opinion on all these labellings. Without loss of generality, we suppose $C:=\{s_1, \dots, s_d\}$ and show that $F =\{a_1, \dots, a_{d+1}\}\cup \{b_1, \dots, b_d\}$.
    
    $(\Rightarrow)$ Assume $C$ is a hitting set for $\{S_1, S_2,..., S_{m}\}$. Since agent $i$ always changes opinion, $F:=\{a_1, \dots, a_{d+1}\}\cup \{b_1, \dots, b_d\}\subseteq N$ is a feasible influencer set if and only if
    \begin{align*}
        \vert M(\ell_k) \cap F\vert > \vert M(\ell_k)^c \cap F\vert, \text{ for } k=1, \dots, \hat{m}.
    \end{align*}
    This holds for all our examples because $\vert M(\ell_k) \cap F\vert \geq d+1$ while $\vert M(\ell_k)^c \cap F\vert \leq d$.
    
    More specifically, for $k=1, \dots, m$, there are $d$ agents in $\{a_1, \dots, a_{d+1}\}$ in each matching set. Further, since agents in $b_1, \dots, b_d$ are associated to the elements in the hitting set $C$, we have that
    \begin{align*}
        \vert M(\ell_k) \cap \{b_1, \dots, b_d\}\vert \geq 1 \text{ and } \\ 
        \vert M(\ell_k)^c \cap \{b_1, \dots, b_d\}\vert \leq d -1.
    \end{align*}
    
    In contrast, all agents in $\{b_1, \dots, b_d\}$ plus one agent from $\{a_1, \dots, a_{d+1}\}$ will be in $M(\ell_k)$ for $k=m+1, \dots, \hat{m} -1$. Therefore, $\vert M(\ell_k)^c \cap F\vert \leq d$ since $\vert F \vert = 2d+1$. Similarly, all (and only) the agents in $\{a_1, \dots, a_{d+1}\}$ belong to the matching set when $k=\hat{m}$.

    $(\Leftarrow)$ Assume $F\subseteq N$ is a feasible set of influencers. We need to show that 
    $$F=\{a_1, \dots, a_{d+1}\}\cup \{b_1, \dots, b_d\},$$
    where the elements $\{s_1, \dots, s_d\} \subseteq \bigcup_{k=1}^m S_k$ associated to agents $b_1, \dots, b_d$ create a hitting set.
    
    Let $F:=A\cup B$ for some $A \subseteq \{a_1, \dots, a_{d+1}\}$ and $B \subseteq \{b_1, \dots, b_{\hat{n}}\}$. However, since $F$ is a feasible influencer set, it must satisfy that
    \begin{align}
        \vert M(\ell_k) \cap F\vert > \vert M(\ell_k)^c \cap F\vert, \text{ for } k=1,\dots, \hat{m}. \label{eq: condition M F Mc F}
    \end{align}
    
    For the $\hat{m}$-th example, we have that $\vert M(\ell_{\hat{m}}) \cap A\vert = \vert A \vert$ and $\vert M^c(\ell_{\hat{m}}) \cap B\vert = \vert B \vert$. So $\vert A \vert > \vert B \vert$ for \eqref{eq: condition M F Mc F} to hold. This ensures that our hitting set remains within budget $d$.

    Now, for $k=m+1, \dots, \hat{m} -1$, $\vert M(\ell_k) \cap B\vert = \vert B \vert$ while $\vert M(\ell_k) \cap A\vert = \Ind{a_{k-m}\in A}$, where $\Ind{a_{k-m}\in A}$ is $1$ if $a_{k-m}\in A$ and $0$ otherwise. Therefore,  
    \begin{align}
        &\vert M(\ell_k) \cap F\vert - \vert M(\ell_k)^c \cap F\vert \notag\\
        =\;&(\Ind{a_{k-m}\in A} + \vert B \vert) - (\vert A \vert - \Ind{a_{k-m}\in A} ) \notag\\
        =\;&2\times\Ind{a_{k-m}\in A} + \vert B \vert - \vert A \vert. \label{eq: ak in F}
    \end{align}
Since \eqref{eq: ak in F} has to be positive with $\vert A \vert > \vert B \vert$, it must be that $a_{k-m}\in A$ for $k=m+1, \dots, \hat{m} -1$. 
Therefore, $A = \{a_1, \dots, a_{d+1}\}$ and $\vert B\vert > d-1$. More so, $\vert B\vert =d$. 

Finally, for $k=1,\dots, m$, we have $\vert M(\ell_k) \cap A\vert = d$. Thus, to satisfy \eqref{eq: condition M F Mc F} we need $$\vert M(\ell_k) \cap F\vert - \vert M(\ell_k)^c \cap F\vert = 2\, \vert M(\ell_k) \cap B\vert - 1 > 0.$$
    This implies that $\vert M(\ell_k) \cap B\vert>1$ for every $k=1,\dots, m$. However, recall that every agent $b_j\in \{b_1, \dots, b_{\hat{n}}\}$ that belongs to $M(\ell_k)$ has an element $s_j\in S_k$ associated with it, making $\{s_1, \dots, s_{d}\}$ a hitting set for $\{S_1, \dots, S_m\}$.
\end{proof}

\begin{figure}[t]
    \centering
    \includegraphics[width=0.75\columnwidth]{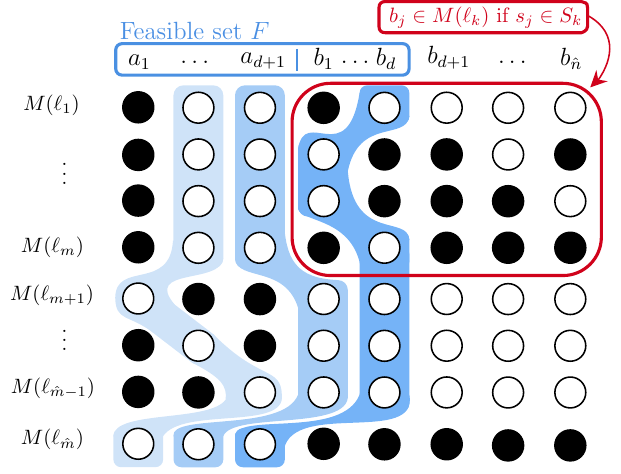}
    \caption{Graphical representation of the matching transformation derived from a Hitting Set instance. Budget is $d=2$ and the input sets are $S_1=\{s_2, s_3, s_4, s_5\}$, $S_2=\{s_1, s_4\}$, $S_3=\{s_1, s_5\}$ and $S_4=\{s_2\}$. One possible hitting set is $C=\{s_1, s_2\}$. So, the set $F:=\{a_1, a_2,a_3\}\cup \{b_1, b_2\}$, consisting of all the auxiliary agents together with agents $\{b_1, b_2\}$ associated with $C$ is a feasible influencer set.}
    \label{fig: reduction as majority matrix}
\end{figure}

\begin{remark}
    Figure \ref{fig: reduction as majority matrix} illustrates the matching transformation $\MM$ built out of a Hitting Set instance following the steps in the proof of Theorem \ref{theo: Hitting set reduction}. We use the same notation as in Figure~\ref{fig: exm majority transformation}, where \majority{white} represents the agents in the matching set. If there is a solution to the Hitting Set problem, we can trace $d+1$ paths from the top to the bottom of $\,\MM$ flowing through the matching sets using just $2d+1$ agents. If we restrict $\MM$ to only the columns of agents in $F$, denoted by $\MM_F$, then more than half of the agents match the oracle in all the rows. Also, the far-right path is made of elements in the non-empty intersection of the hitting set with the sets in $\{S_1, \dots, S_m\}$.
\end{remark}

The Hitting Set problem is known to be \textbf{NP}-complete (see, e.g., \cite{dasgupta_algorithms_2006}). Yet, we are also interested in the Randomised Polynomial-time (\textbf{RP}) complexity class. Algorithms in \textbf{RP} run in polynomial time in the input size. Also, they guarantee no false positives and a false-negative probability of less than $\nicefrac{1}{2}$. 

In our context, false positives occur when an influencer set is retrieved but it is inconsistent with at least one of the examples in $(\ell_k)_{k=1}^m$ (i.e., the agent changes opinion when predicted not to, or vice versa). A false negative occurs if the algorithm claims there is no feasible influencer set for the given sample when such a set actually exists. Therefore, the reduction from Theorem \ref{theo: Hitting set reduction} yields that, unless $\mathbf{NP}=\mathbf{RP}$, no polynomial-time algorithm is capable of PAC learning feasible influencer sets in networks with majority dynamics.

\begin{theorem}
   Let $f^+$ be the class of majority dynamics protocols in social networks formed by $n+1$ agents who can hold binary opinions in $L$. If there exists an algorithm $\lL$ that is a CHF such that, for every $G_i \subseteq N$, distribution $\dD$ over $L^N$ and error parameter $0<\varepsilon<1$, $\lL$ runs in polynomial time for $n:=\vert N\vert$ and $\nicefrac{1}{\varepsilon}$ and, with probability of at least $\nicefrac{1}{2}$, outputs an influencer set $F\subseteq N$ satisfying $\PP_{\ell \sim \dD}(f^+(\ell,F) = f^+(\ell,G_i)) \geq 1 - \varepsilon$, then $\mathbf{NP}=\mathbf{RP}$.
   \label{theo: majority NP vs RP}
\end{theorem}

\begin{proof}[Proof (Outline)]
   Since $\lL$ can find a feasible influencer set in polynomial time for any distribution of network labellings, for any Hitting Set instance with $m$ sets and budget $d$, consider the uniform distribution over network labellings restricted to the examples built as in the proof of Theorem \ref{theo: Hitting set reduction}. Over this distribution, taking $\varepsilon\leq (m+d+3)^{-1}$ implies $\lL$'s solution perfectly fits the oracle predictions, making it a CHF. Further, this solution can be translated back to a Hitting Set instance and into any $\NP$-hard problem. As an $\RP$-algorithm, $\lL$ establishes the link between the classes.
\end{proof}

\begin{cor}
    Assuming $\mathbf{NP}\neq\mathbf{RP}$, there is no polynomial-time algorithm for learning the influencers of an agent in a social network with majority dynamics. \label{cor: assume NP neq RP}
\end{cor}

\section{A Heuristic for Majority Dynamics} 
\label{sec: Waterfall algorithm} 

Theorem~\ref{theo: majority NP vs RP} and Corollary~\ref{cor: assume NP neq RP} show that retrieving, even approximately, the structure of a social network just from observing its majority dynamics is computationally intractable. This, however, should not prevent us from designing practical heuristics. Recall that our first naive idea was to return the agents in the intersection of all matching sets as a feasible influencer set. While unanimity dynamics guarantee a non-empty intersection, that is not the case for majority dynamics. Thus, we relax our approach to consider agents outside some matching sets as feasible influencers, provided that these are never the strict majority. 


Formally,  given an $m\times n$ matching transformation $\MM$, an oracle prediction $f^+(\ell_k, G_i)$ for $k=1,\dots, m$, and a subset $F\subseteq N$, the rows of the restricted matrix $\MM_F$ fall into:

\begin{enumerate}
    \item \textbf{Consistent (C):} If $\vert M(\ell_k) \cap F\vert > \vert M(\ell_k)^c \cap F\vert$.
    \item  \textbf{Barely consistent (BC):} If $\vert  M(\ell_k) \cap F\vert - \vert  M(\ell_k)^c \cap F\vert=1$ and $f^+(\ell_k,G_i)$ is $\mathit{True}$.
    \item \textbf{Consistent tie (CT):} If $\vert M(\ell_k) \cap F\vert = \vert M(\ell_k)^c \cap F\vert$ and $f^+(\ell_k,G_i)$ is $\mathit{False}$.
    \item \textbf{Inconsistent tie (IT):} If $\vert M(\ell_k) \cap F\vert = \vert M(\ell_k)^c \cap F\vert$ and $f^+(\ell_k,G_i)$ is $\mathit{True}$.
    \item \textbf{Inconsistent (I):} If $\vert M(\ell_k) \cap F\vert < \vert M(\ell_k)^c \cap F\vert$.
\end{enumerate}

As we append an agent to $F$, each row becomes more consistent or inconsistent depending on whether the agent is in $M(\ell_k)$ or $M(\ell_k)^c$. Yet, a row in a consistent state cannot become inconsistent (or vice versa) without going through \textbf{CT} if $f^+(\ell_k, G_i)$ is $\mathit{False}$, or \textbf{BC} if it is $\mathit{True}$. Moreover, our learner $\lL$ finds a feasible influencer set $F^* \subseteq N$ when no rows in $\MM_{F^*}$ are in states \textbf{I} or \textbf{IT}.

All inconsistencies must be resolved simultaneously upon adding the final agent, so no labelling is prioritised over another. For a given $F\subseteq N$, we say a labelling $\ell \in (\ell_k)_{k=1}^m$ \emph{needs rescuing} if it does not admit adding any agent from $M(\ell)^c$ (i.e., its state is \textbf{I, IT, CT} or \textbf{BC}). Our heuristic employs a greedy strategy when adding agents to $F$. It picks the agent belonging to the most matching sets of labellings that need rescuing. We call this the \texttt{Waterfall} algorithm. Its pseudocode is shown in Algorithm~\ref{alg: Waterfall algorithm}, and the implementation used in our experiments will be made publicly available upon publication at \href{https://github.com/lf-estrada/Waterfall-algorithm.git}{https://github.com/.../Waterfall-algorithm.git}

\begin{algorithm}[ht]
\caption{Waterfall algorithm}
\label{alg: Waterfall algorithm}
\textbf{Input:} An oracle prediction $f^+(\ell_k,G_i)$ for $k=1,\dots, m$ and its $m\times n$ matching transformation $\MM$.\\
Assume $\cap_{k=1}^m M(\ell_k) = \emptyset$ and let $N:=\text{columns}(\MM)$.\\
\textbf{Output:} A set of influencers $F$ that makes all rows in $\MM_{F}$ consistent or a consistent tie.
\begin{algorithmic}[1]
\If{$\{\}$ is a feasible set} \textbf{return} $\{\}$.
\EndIf
\For{$s \in N$}
\State Initialise $F = \{s\}$.
\While{$F \neq N$}
\State Retrieve state vector $\bar{\rr}$ for the rows of $\MM_{F}$.
\If{$\exists \; k \leq m : \bar{\rr}_k \in \{ I, IT\}$}
\LComment{Some labellings need rescuing.}
\State Set $R:= \{ k \leq m : \bar{\rr}_k \in \{\textbf{I, IT, CT, BC}\}\}$
\State Choose an agent $j\in N\mysetminus F$ from \\ $\argmax_{j\in N\mysetminus F} \sum_{r\in R} \Ind{j\in M(\ell_r)}$.
\State Add $j$ to $F$.
\Else
\State \textbf{return} $F$.
\EndIf
\EndWhile
\vspace{0.05cm}
\LComment{Initialise waterfall on another source if the while loop ends and there are still inconsistencies.}
\EndFor
\If{$N$ is a feasible set} \textbf{return} $N$.
\Else{ \textbf{return} ``Terminates without finding a Waterfall."} 
\EndIf
\end{algorithmic}
\end{algorithm}

\begin{proposition}
    \label{prop: complexity if found}
    The \texttt{Waterfall} runs in $\oO(m\, n^3)$ time for any $m\times n$ matching matrix $\MM$ and oracle prediction.
\end{proposition}

 \begin{remark}
    To adapt the \texttt{Waterfall} to any $\tau$-margin protocol, $\tau \geq 0$, we redefine the labellings that need rescuing as those where $\vert M(\ell_k)\cap F\vert - \vert M(\ell_k)^c\cap F\vert \leq \tau$. This preserves the intuition of requiring an additional agent from the matching set; otherwise, the labelling would be inconsistent. The margin is also incorporated into any feasibility check. 
\end{remark}

\subsection{Theoretical Guarantees}
\label{sec: correctness Waterfall}

The \texttt{Waterfall} performs consistency checks in lines 1, 6 and 14, before returning a feasible influencer set. Still, it may incur false-negative errors, in the sense of incorrectly concluding there is no feasible influencer set for the given sample when, in fact, such a set exists. This is because, like any greedy approach, the \texttt{Waterfall} is prone to follow suboptimal routes. For instance, if multiple agents rescue the same number of labellings, the algorithm may select one outside the true influencer set $G_i$, though other feasible sets may exist. So, when can we guarantee the \texttt{Waterfall} will not add too many ``incorrect'' influencers?

\begin{proposition}
    If a feasible influencer set $F\subseteq N$ exists and $\vert F \vert \leq 2$, then the \texttt{Waterfall} will find it.
    \label{prop: feasible pair}
\end{proposition}

\begin{proposition}
    If a feasible influencer set $F\subseteq N$ exists and the \texttt{Waterfall} goes over $F\mysetminus \{j\}$, for any $j\in F$, then it will find a feasible influencer set of size $\vert F\vert$. 
    \label{prop: only I1 to make a feasible set}
\end{proposition}

Propositions \ref{prop: feasible pair} and \ref{prop: only I1 to make a feasible set} ensure that the \texttt{Waterfall} finds a solution if it iterates over any subset of agents that is one agent away from being feasible. This is because any ties occur among agents who rescue all the remaining labellings. Therefore, the failure probability corresponds to the likelihood of omitting all subsets that are one agent short.

Our algorithm resembles the \texttt{Greedy} cover algorithm in \cite{lovasz_ratio_1975}, which was pivotal for proving that the ratio of optimal integral and fractional covers of a hypergraph $G$ is at most $1+\log \deg(G)$, where 
$\deg(G)$ is the hypergraph's maximum degree. \texttt{Greedy}'s approximation is shown to be optimal in worst-case scenarios in \cite{feige_threshold_1998}. Further, \cite{arpino_greedy_2023} characterised the integrality gap, up to multiplicative constants, for the average-case of randomised Hitting Set instances. These instances assume that each set includes the given element independently with probability $p$, and the characterisation relies on a combinatorial analysis in dense and sparse regimes.

Our \texttt{Waterfall} traverses the initialisation space of \texttt{Greedy}, mitigating the dependencies introduced by the first pick. Also, it operates on the problem of finding a feasible influencer set, which is a generalisation of the Hitting Set. For instance, there is no monotonicity to rely on, so adding more elements does not necessarily lead to a valid solution. Think of the case where taking all agents in $N$ as agent $i$'s influencers conflicts with the oracle’s prediction because the global majority differs from the majority in $G_i$. Still, the feasibility checks built into the \texttt{Waterfall} prevent it from returning such false positives. 

\subsection{Experiments}
We run the \texttt{Waterfall} on four types of random networks: Erdős–Rényi (ER), Watts–Strogatz (WS), Regular Graph (RG) and Barabási–Albert (BA); generated using NetworkX. We vary the network size $n$ and number of examples $m$ between $10$ and $50$. For each $(n,m)$ pair, we generate $50$ labelling samples of size $m$ and $40$ random networks of size $n$ per sparsity regime $p\in\{0.1, 0.25, 0.5, 0.75, 0.9\}$. Networks are divided evenly among the graph types. We build the oracle prediction and matching transformation for each agent in the network and the labelling sample.


For fair comparison, the sparsity regime $p$ governs the structure across network types. It specifies the edge probability in ER and the rewiring probability in WS. We use $2 + p\,(\frac{n}{2} - 1)$, rounded to the nearest even integer, to control the number of neighbours in WS, the degree in RG, and how many edges each new node attaches to in BA. Exact generation functions are detailed in the Appendix.
 
Figure~\ref{fig: Waterfall algorithm fails} presents the False Negative Rate (FNR) of the \texttt{Waterfall} on inputs where a solution always exists. Our results show performance is agnostic to network topology. Breakdowns by network type and sparsity regime are available in the supplementary material. Instead, performance depends on the ratio of labelling samples to network size. Error rates are confined within a cone-shaped region and increase towards its interior. Intuitively, underdetermined cases fall below the cone, where extra agents enable incorrect yet feasible influencer sets. Conversely, the cone's upper bound reflects how additional examples reduce the likelihood of ties. Overall, the \texttt{Waterfall} succeeded $98.13\%$ of the trials, with a mean FNR of $2.08\times10^{-3}$ and variance $5.40\times 10^{-6}$. 
\begin{figure}[ht]
    \centering
    \includegraphics[width=0.75\columnwidth]{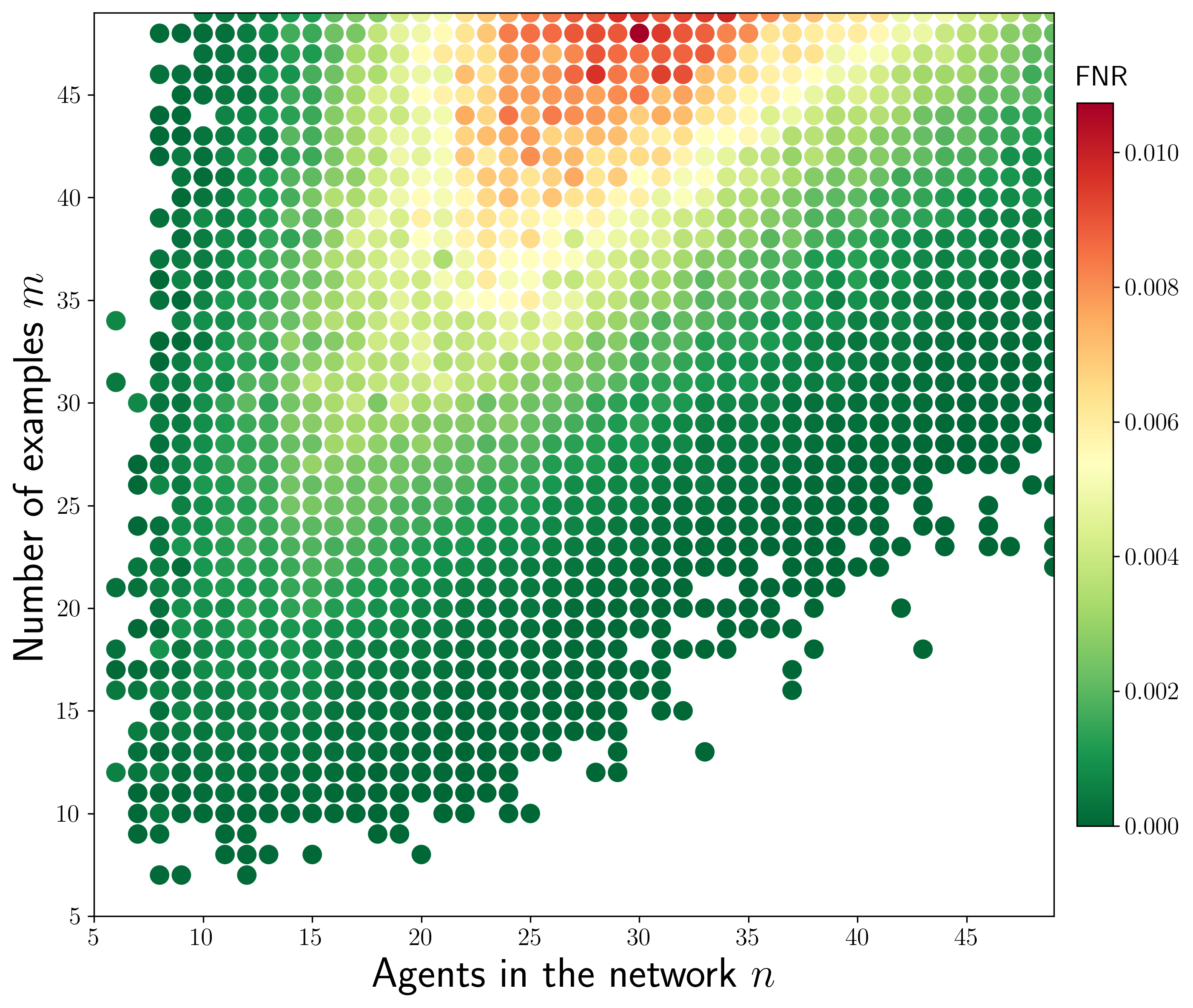}
    \caption{False Negative Rate (FNR) for the \texttt{Waterfall}. Each grid point $(m,n)$ tested 40 networks with~$n$ agents per density value $p\in\{0.1, 0.25, 0.5, 0.75, 0.9\}$ on 50 shared labelling samples of size $m$. The \texttt{Waterfall} runs in parallel across the agents in each network. For $n \leq 5$, we exhausted all possible networks and labelling samples without errors.}
    \label{fig: Waterfall algorithm fails}
\end{figure} 

\section{Discussion}
\label{sec: discussion}

We studied the problem of learning the structure of a social network governed by threshold-based opinion dynamics. We presented Consistent Hypothesis Finders (CHFs) that run in polynomial time for the cases of unanimity and all-but-$\kappa$ update rules. For majority dynamics, we proved that finding feasible influencer sets is $\NP$-complete. We designed a greedy polynomial-time heuristic to tackle this problem, which achieved over $98\%$ success rate on our tests in random networks, with false negatives contained within linear bounds of the parameters $n$ and $m$. These results align with the expected growth rate of the number of examples required for PAC learning shown in Equation $\eqref{eq: how many examples}$. 

Future work includes deriving closed-form formulas for the probability of false negatives in the \texttt{Waterfall}, as well as testing its performance on real-world networks. Questions about how to incorporate noise into the model remain open. For example, consider probabilistic opinion changes or inaccurate opinion reports, which are common in clinical trials. Similarly, it is unclear how the heuristic performs when the diffusion protocol is also unknown, and even more so, when there is agent heterogeneity.


\backmatter

\section*{Research Funding Statements}

Dmitry Chistikov is supported in part by the Engineering and Physical Sciences Research Council [EP/X03027X/1]. Luisa Estrada acknowledges the support of the Engineering and Physical Sciences Research Council through the Mathematics of Systems II Centre for Doctoral Training at the University of Warwick [EP/S022244/1]. Paolo Turrini acknowledges the support of the Leverhulme Trust for the Research Grant RPG-2023-050 and the TAILOR Connectivity Fund (Agreement 29).

\begin{appendices}

\section{Full Proofs}
\label{sec: full proofs}
\begin{proof}[Proof of Proposition \ref{prop: CHF unanimity}]
     Let $L_a:=\{\ell_{a}\}_{a=1}^{m_{a}}$ and $L_b:=\{\ell_{b}\}_{b=1}^{m_b}$ be, respectively, the always-changing and never-changing subsets of the sample $\{\ell_k\}_{k=1}^m$, for some $m_a+m_b = m$. Since the network has unanimity dynamics, all influencers of agent $i$ must belong to the matching sets of the always-changing labellings. Therefore, $G_i\subseteq \bigcap_{\ell\in L_a} M(\ell)$, which also implies that $\bigcap_{\ell\in L_a} M(\ell)\neq \emptyset$ if a solution exists. 

     On the other hand, for never-changing labellings, there is no restriction on how many influencers may be in $M(\ell_b)^c$, $\ell_b\in L_b$, as long as it is not all of them (i.e., $G_i\nsubseteq \bigcap_{\ell\in L_a} M^c(\ell)$). Moreover, for any superset $F\supseteq G_i$, if $G_i$ has at least one agent in the matching set of $\ell_{b}$, then so does $F$. Thus, any superset $F$ such that $G_i \subseteq F\subseteq \bigcap_{\ell\in L_a} M(\ell)$ will be a feasible influencer set. In particular, an algorithm that returns $\bigcap_{\ell\in L_a} M(\ell)$ qualifies as a CHF with time complexity $\oO(m_a\, n^2)$, further capped by $\oO(m\, n^2)$ for any oracle's prediction.
\end{proof}

\begin{proof}[Proof of Proposition \ref{prop: always changing all but k}]
    We show that a feasible set of influencers of size at most $\kappa+1$ must exist. Consequently, an algorithm that generates all the subsets of agents of size $\kappa+1, \kappa, \kappa-1, \dots$, checks whether they are contained within each $M(\ell_k)^c$, and returns the first one that is not contained in any of them, constitutes a CHF. Furthermore, since $\kappa$ is fixed, this CHF has a time complexity of $\oO(m \, n^{\kappa+2})$ 

    The first case is when the influencer set of agent $i$, $G_i\subseteq N$, is such that $\vert G_i \vert \leq \kappa + 1$. But then, the exhaustive search over all subsets of size $\kappa+1$ or less will eventually reach $G_i$ and retrieve it as a feasible set. 
    
    On the other hand, when $\vert G_i \vert > \kappa + 1$, it still cannot have more than $\kappa$ agents outside each matching set. In an always-changing sample, the agents in the matching set disagree with agent $i$ (i.e., $j\in M(\ell_k)$ if and only if $\ell_k(j)=\neg \phi$). Therefore, for every $\ell \in (\ell_k)_{k=1}^m$, any subset $F\subseteq G_i$ satisfies  
    $$\kappa \geq \vert M(\ell)^c \cap G_i\vert \geq \vert M(\ell)^c \cap F\vert,$$ ensuring that $\sum_{j\in F} \Ind{\ell(j)=\phi} \leq \kappa$. However, a feasible influencer set must have at least one agent who disagrees with agent $i$ for it to change opinion. This is guaranteed for any subset $F\subseteq G_i$, $\vert F \vert=\kappa +1$, as 
    $$\vert M(\ell) \cap F\vert \geq \vert F \vert - \vert M(\ell)^c \cap F\vert \geq 1.$$ 
    Thus, our CHF returns the first subset of $G_i$ with $\kappa+1$ agents it finds during its search.
\end{proof}

\begin{proof}[Proof of Theorem \ref{theo: majority NP vs RP} ]
    Suppose that there exists a polynomial-time algorithm $\lL$ to learn the influencers of agents in social networks with majority dynamics. Given a budget $d>0$ and a set of sets $\{S_1, \dots , S_m\}$, take the networks with $n+1$ agents, where $n=\vert \bigcup_{k=1}^m S_k \vert + d +1$. Pick a target agent $i$ and construct a set of $m+d+2$ example network labellings as in the proof of Theorem \ref{theo: Hitting set reduction}. Set the error parameter to
    
    $$\varepsilon = \frac{1}{m+d+3}.$$
    
    Run $\lL$ to obtain with probability $\nicefrac{1}{2}$ a set of influencers $F\subseteq N$ for agent $i$, if a feasible set of influencers exists. If no such set exists, $\lL$ halts without returning anything. The algorithm runs in polynomial time for $n$ and $\nicefrac{1}{\varepsilon} = m+d+3$, so it will be polynomial in the size of $\bigcup_{k=1}^m S_k$, the number of sets in $\{S_1, \dots , S_m\}$ and the budget $d$. 

    Now, consider a uniform distribution $\dD$ over the examples and recall that agent $i$ always changes opinion. Therefore, the error $\varepsilon$ of $F$ with respect to $\dD$ is
    \begin{align*}
        \varepsilon &= \sum_{\ell \in \{1,0\}^N} \Ind{f^+(\ell,F)} \dD(\ell)\\
        &= \frac{1}{m+d+2} \sum_{\ell \in (\ell_k)_{k=1}^{m+d+2}} \Ind{f^+(\ell,F)}.
    \end{align*}
    However, because $\varepsilon < (m+d+2)^{-1}$, we have that $f^+(\ell,F)$ is $\mathit{True}$ for all the labellings in our examples. Thus, $F$ is a feasible set of influencers, and from Theorem \ref{theo: Hitting set reduction}, we can reconstruct a hitting set for $\{S_1, \dots, S_m\}$. Moreover, since the Hitting Set problem is \textbf{NP}-complete, $\lL$ would be capable of solving any \textbf{NP}-hard problem while remaining \textbf{RP} itself. \\
\end{proof}

\section{In-Depth Analysis of the \texttt{Waterfall} Algorithm}
\label{sec: deep Waterfall}
When there are no agents that belong to the matching sets of all the labellings in our sample (i.e., $\cap_{k=1}^m M(\ell_k)=\emptyset$), we need to strike a balance when picking agents from inside and outside the matching set of each labelling. There must never be strictly more agents that do not match the oracle prediction, and equality can only hold for non-changing labellings. We introduce the concepts of streams and waterfalls to keep track of this and provide a graphical intuition of the \texttt{Waterfall}.

\subsection{Streams and Waterfalls}
\label{sec: streams and waterfalls}
\begin{definition}
    Let $\MM$ be the $m\times n$ matching transformation of a labelling sample $(\ell_k)_{k=1}^m$. Then, the sequence $\ss :=(s_k)_{k=1}^m$ is a \emph{stream} over $\mathbf{M}$ if $s_k \in M(\ell_k) \subseteq N$ for $k=1, \dots, m$. Moreover, we say that a collection $W=\{\ss^1, \dots , \ss^w\}$, $w\geq 1,$ of streams over $\mathbf{M}$ is a \emph{waterfall} of size $w$ if $s^u_k \neq s^v_k$ for $k=1, \dots, m$ and any $u,v \in \{1, \dots, w\}, \; u\neq v$. 
\end{definition}

\begin{remark}
    We use $W_k:=\{s_k^1, \dots, s_k^w\}$ as a shorthand for the agents that belong to a stream on the $k$-th row of $\MM$. 
\end{remark}

\begin{definition}
    \label{def: waterfall ambit}
    The \emph{ambit} of a waterfall $W$ of size $w\geq 1$ over an $m\times n$ matching transformation $\MM$ is the minimal set of agents that can form all the streams in $W$. It is denoted by
    $$\AW:= \cup_{v=1}^w \cup_{k=1}^m s_k^v.$$
\end{definition}

An easy way to avoid confusion between the size and ambit of a waterfall is to remember that \underline{s}ize refers to the \underline{s}treams, while \underline{a}mbit refers to the \underline{a}gents. 

We allude to streams and waterfalls as they can be visualised as paths flowing down the \majority{white} nodes of a matching transformation $\MM$ without overlapping. We present how to use these devices to verify if a subset $F\subseteq N$ is a feasible influencer set.

\begin{proposition}
    \label{prop: iif condition}
    Let $\MM$ be an $m\times n$ matching transformation. Then, a subset $F\subseteq N$ is a feasible set of influencers over $\MM$ if and only if a waterfall $W$ of size $w\geq~\ceil{\vert F\vert/2}$ can be built over $\MM_F$. Further, if $\vert F \vert$ is even and $w=\vert F \vert/2$, we allow $\vert M(\ell_k) \cap F\vert =\vert F \vert/2$ only when $f^+(\ell_k,F)$ is $\mathit{False}$.
\end{proposition}

\begin{proof}
    For $(\Leftarrow)$, we have that $\vert M(\ell_k) \cap F \vert \geq \ceil{\vert F \vert/2}$ for $k=1,\dots, m$, since waterfall $W$ has at least that many streams flowing through $F$. Consequently, $\sum_{j\in F} \MM_{kj}\geq 0$. If $\sum_{j\in F} \MM_{kj}> 0$, this means that there is a strict majority of agents in $F$ whose opinion in $\ell_k$ matches the output opinion of agent $i$ predicted by $f^+(\ell_k, G_i)$. When $\sum_{j\in F} \MM_{kj}=0$, there is a tie, which can only occur on non-changing labellings. So, in either case, $F$ is a feasible set for $\ell_k$.
    
    For $(\Rightarrow)$, we assume that $F$ is a feasible influencer set, so agent$~i$ agrees with at least half of the agents in $F$ matched the prediction after one opinion diffusion step of $f^+(\ell^+_k, G_i)$, for $k=1,\dots, m$. and Therefore, $\vert M(\ell_k) \cap F \vert \geq \vert F \vert / 2$. Ties occur if $\vert M(\ell_k) \cap F \vert = \vert F \vert/2$, but since $F$ is a feasible influencer set, agent $i$ does not change opinion in this case, and $M(\ell_k)$ are the agents who agree with agent $i$ in $\ell_k$. Thus, for each $\ell_k$, we can: 
    \begin{enumerate}
        \item Enumerate the agents in $M(\ell_k)\cap F:=\{k_1, \dots, k_{w_k}\}$, for some $w_k\geq \vert F \vert/2$.
        \item Pick $w:=\min_{k\in \{1, \dots, m\}} w_k \geq \vert F \vert/2$.
        \item Build a stream $\ss^v=(k_v)_{k=1}^{m}$ over $\MM_F$ for every\\ $v=1,\dots, w$.
    \end{enumerate} 
    Since none of these streams overlap, together they form a waterfall $W=\{\ss^1, \dots, \ss^w\}$ of size $w\geq \vert F \vert/2$.
\end{proof}

\begin{proof}[Proof of Proposition \ref{prop: complexity if found}]
    Let $F:=\AW$ be the ambit of the waterfall and assume the \texttt{Waterfall} found a waterfall $W$ with the right ratio between $\vert\alpha(W)\vert$ and $w:=\vert W \vert$ according to Proposition \ref{prop: iif condition}. Given a source node $s\in N$, the algorithm has to update the state vector $\rr$ at most $n$ times to build $W$. At each update, it computes the row sum of $\MM_{\AW}$ and classifies it as C, CT, BC, IT or I. Next, it performs a column-sum over $\MM_{\alpha(W)^c}$ to find how many labellings each agent rescues whilst keeping track of the set of top-rescuing agents. An agent from this set is selected in unit time.\footnote{\; The selection can be refined by adding filters. This increases the complexity by at most a constant factor (number of filters).} These steps are sequential and take $\oO(m\, n)$ time. Since up to $n$ agents can be added to $W$, building a waterfall from a source takes $\oO(m\, n^2)$ time. Iterating over all the possible source nodes $s\in N$ yields $\oO(m\, n^3)$ as claimed.
\end{proof}

\begin{proof}[Proof of Proposition \ref{prop: feasible pair}]
    Let $\MM$ be the matching transformation built from the labelling sample $(\ell_{k})_{k=1}^m$. When $\vert F\vert=0$, the \texttt{Waterfall} finds it because the first check it does is whether $\{\}$ is a feasible set of influencers. When $\vert F\vert=1$, we have that $\cap_{k=1}^m M(\ell_k) \neq \emptyset$. This means that for any agent $j\in \cap_{k=1}^m M(\ell_k)$, the algorithm will return $\alpha(W)=\{j\}$ as a feasible set of influencers once $j$ is the source node. Finally, let $\vert F\vert=\{j_1, j_2\}$, for $j_1, j_2\in N$, $j_1\neq j_2$. Without loss of generality, assume the \texttt{Waterfall} goes over agent $j_1$ as a source node and that the state vector of $\MM_{\{j_1\}}$ has $m_{1}<m$ labellings that need rescuing. If $F=\{j_1, j_2\}$ is a feasible set, then $j_2$ is a perfect match for $\{j_1\}$ and must rescue $m_1$ labellings. Any other agent tied with $j_2$ must also rescue $m_1$ labellings, so they will also be a perfect match for $\{j_1\}$. Regardless of which of these agents the \texttt{Waterfall} picks, the next validation point will find a feasible set.
\end{proof}

\begin{proof}[Proof of Proposition \ref{prop:  only I1 to make a feasible set}]
    Notice that one additional agent can only turn a labelling $\ell\in L^N$ from inconsistent to consistent if its margin is $M(\ell)\cap \AW\vert - \vert M(\ell)^c\cap \AW\vert\geq -1$. Additionally, we require equality if $f^+(\ell, G_i)$ is $\mathit{False}$. Let $\bar{\rr}$ be the state vector associated to $\MM_{\AW}$, whose labellings that need rescuing are stored in 
    $$R:=\{k\leq m : \bar{\rr}_k \in \{ I_1, IT, CT, BC\}\}, \; \vert R\vert = m_\alpha <m.$$ 
    If $\AW$ has a perfect match, then all the agents in 
    $$\argmax_{j\in \alpha(W)^c}\sum_{r\in R} \Ind{j\in M(\ell_r)},$$
    rescue $m_\alpha$ labellings. Regardless of which agent the \texttt{Waterfall} picks, the next validation point will find they form a feasible set as the new state vector will either have positive entries or maybe none where $f^+(\ell, G_i)$ is $\mathit{False}$.
\end{proof}

\subsection{Why Do We Call It the \texttt{Waterfall}?}

\begin{example}
    Suppose the \texttt{Waterfall} receives as input the $m\times n$ matching transformation $\MM$ and the oracle prediction for $(\ell_k)_{k=1}^m$ shown in Figure \ref{fig: Waterfall example 1}. Thus, we have a social network with $n=6$ agents plus the target agent $i$ and a sample of $m=8$ examples. We initialise our waterfall on the source node $i_1 \in N$. We bring an agent to the left of the state vector $\bar{\rr}$ to show we added it to $\alpha(W)$. The state of each $\bar{\rr}_k$ depends on the oracle prediction for $\ell_k$ and the sign of $\vert M(\ell_k)\cap \alpha(W)\vert - \vert M(\ell_k)^c\cap \AW\vert$.
    
    For each row in the restricted matrix $\MM_{\AW}$, we number the entries that belong to the majority set from left to right. Then, we let the streams in $\alpha(W)$ flow through the nodes with the same value. The size $w:=\vert W \vert$ of our waterfall is the minimax of these values.\footnote{\; \emph{Minimax:} Minimum of the maximums.}
    
    In general, our assumption of $\cap_{k=1}^m M(\ell_k)=\emptyset$ implies we cannot form a stream with a single node. This is seen in Figure \ref{fig: Waterfall example 1} with $\vert\alpha(W)\vert=1$ and $w = 0$. In it, $\alpha(W)$ is unfeasible for the labellings $\ell_1, \ell_2$ and $\ell_3$, although only $\ell_6$ and $\ell_7$ do not need rescuing. Here, agent $i_5$ rescues the most labellings. We add it to $\alpha(W)$ to obtain Figure \ref{fig: Waterfall example 2}. 
    
    We have one stream over two agents if $\alpha(W)=\{i_1, i_5\}$, so it has the potential to be a feasible set. Nevertheless, $\ell_3$ and $\ell_5$ have inconsistent ties. Thus, we need to build another stream. In this step, agent $i_2$ rescues the most labellings, so we add it to $\alpha(W)$, which yields Figure \ref{fig: Waterfall example 3}. Finally, because $\alpha(W)=\{i_1, i_5, i_2\}$ satisfies that $\ceil{\vert \alpha(W)\vert/2}=w=2$, Proposition \ref{prop: iif condition} tells us that $\alpha(W)$ is a feasible set of influencers. This is further verified by the lack of inconsistent states in $\bar{\rr}$. Consequently, the \emph{\textbf{else}} condition in the \texttt{Waterfall} is met, and it returns $\alpha(W)$.

\begin{figure}[H]
    \centering
    \begin{subfigure}[t]{0.9\linewidth}
    \centering
    \includegraphics[width=0.6\linewidth]{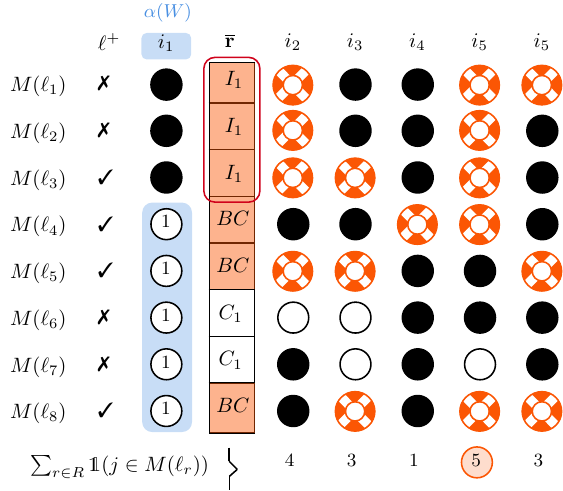}
    \caption{Agent $i_1$ as source node in the \texttt{Waterfall}, $\vert \alpha(W)\vert = 1, \; w = 0.$}
    \label{fig: Waterfall example 1}
    \end{subfigure}

    \begin{subfigure}[t]{0.9\linewidth}
    \centering
    \includegraphics[width=0.6\linewidth]{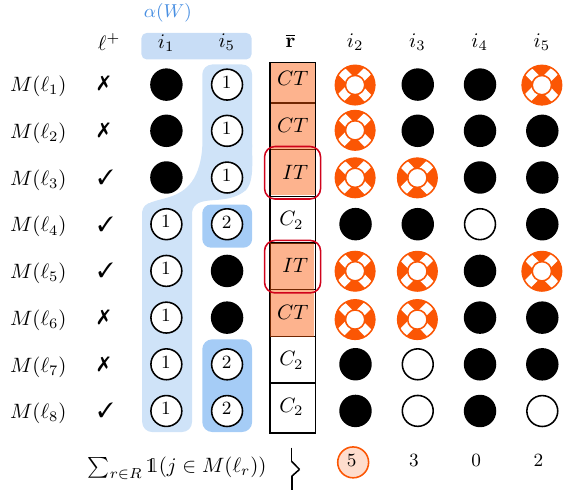}
    \caption{The first stream emerges in the \texttt{Waterfall}, $\vert \alpha(W)\vert = 2, \; w = 1.$}
    \label{fig: Waterfall example 2}
    \end{subfigure}
\end{figure}
\begin{figure}[H]\ContinuedFloat
    \centering
    \begin{subfigure}[t]{0.9\linewidth}
    \centering
    \includegraphics[width=0.6\linewidth]{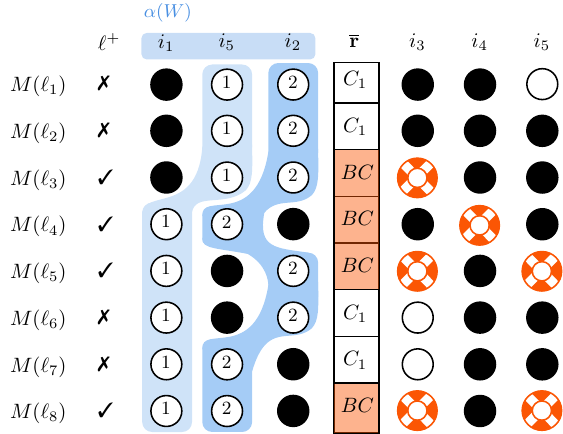}
    \caption{A feasible influencer set $\alpha(W)$ is found as per Proposition \ref{prop: iif condition}, $\vert \alpha(W)\vert = 2, \; w = 1.$}
    \label{fig: Waterfall example 3}
    \end{subfigure}
    
    \caption{Example of the \texttt{Waterfall}. Agents depicted with \protect\majority{white} belong to the majority set $M(\ell)\subseteq N$ and with \protect\majority{black} if they belong to $M(\ell)^c$ for $\ell$ in the labelling sample $(\ell_{k})_{k=1}^m$. The oracle prediction is such that \cmark indicates $f^+(\ell, G_i)$ is $\mathit{True}$ and \xmark indicates it is $\mathit{False}$. Agents in $\alpha(W)$ are on the left of the state vector $\rr$ with their streams highlighted. In $\rr$, we mark inconsistent labellings and labellings that need to be rescued with \protect\unfeasible and \protect\rescue, respectively. Finally, we use \protect\rescuer to represent the labellings that each agent in $N\mysetminus \alpha(W)$ rescues.}
    \end{figure}
\end{example} 

\subsection{Implementing the Tie-Breaking Subroutine}
\label{sec: tie-breaking subroutine}

More robust versions of the \texttt{Waterfall} can be achieved by implementing tie-breaking subroutines. For example, Algorithm \ref{alg: filter subroutine} presents a subroutine that recalculates the labellings that need to be rescued by the tied agents, filtering out the most consistent labels from the last tie. That is, it first omits the \textbf{BC} states and then, if ties persist, it progressively ignores the labellings that satisfy $\vert M(\ell) \cap F \vert - \vert M(\ell)^c \cap F \vert \geq \tau$, raging $\tau$ from $\tau = 0$ down to $\tau = -|F|$, or until there are no more ties. Notably, as seen in Figures \ref{fig: size ambit single filter} and \ref{fig: size ambit multi filter}, this subroutine ensures the feasible influencer set for $N=[4]$ is at most $\vert G_i \vert$, unlike the single-filter version. Yet, this comes with an increased complexity determined by the number of filters, which sets how many times the algorithm recalculates the labellings rescued by each agent in $\alpha(W)^c$. Since filters are capped by the sample size, the filtered version of the \texttt{Waterfall} terminates in $\mathcal{O}(m^2 \, n^3)$ time.

\begin{algorithm}
\caption{Filters subroutine}
\label{alg: filter subroutine}
\begin{algorithmic}[1]
\State \textbf{Input:} An oracle prediction $(\ell^+_k)_{k=1}^m$, a matching transformation $\MM$ and the current waterfall ambit $\AW$.
\State Initialise the set of available agents $A = \alpha(W)^c$.
\State Calculate the margin vector $\Delta$ s.t. $\Delta_{k}=\sum_{j\in \AW} \MM_{k,j}$ for $k=1, \dots, m$.
\State Initialise $filter$ as the values in $\Delta$ less or equal than $1$.
\While{$filter \neq \emptyset$}
\State Initialise the remaining agents' vector $\mathbf{r}= [0]\times \vert A\vert$.
\LComment{Counter for labellings rescued per agent.}
\For{$j \in A$}
\For{$k=1,\dots, m$}
\If{$\MM_{k,j} = 1$}
\Comment{Agent $j\in M(\ell_k)$}
\If{$\Delta_k \leq max(filter)$}
\If{$\Delta_k = 1$ and $\ell_{k}^+$}
\LComment{Extra check for BC state.}
\State Update $\rr_j = \rr_j + 1$.
\Else \; Update $\rr_j = \rr_j + 1$.
\EndIf
\EndIf
\EndIf
\EndFor
\EndFor
\vspace{0.1cm}
\State Update $A = \argmax_{j\in A} \rr_j$.
\If{$\vert A \vert = 1$} \textbf{return} $A$. \Comment{No ties.}
\Else \; Update $filter = filter \mysetminus max(filter)$.
\EndIf
\EndWhile
\State \textbf{return} any agent $j\in A$.
\end{algorithmic}
\end{algorithm}

\begin{figure}[H]
    \centering
   \includegraphics[width=0.75\columnwidth]{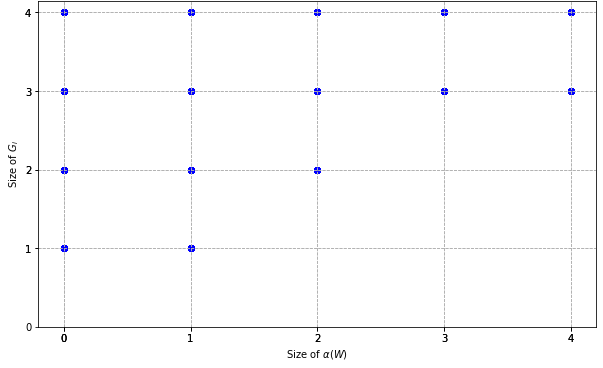}
   \caption{Number of agents in the feasible set $\alpha(W)$ compared to the true influencer set $G_i$ for networks with at most $5$ agents. The \texttt{Waterfall} in its single-filter version, with a uniformly at random tie-breaking criterion.}
   \label{fig: size ambit single filter}
\end{figure}

\begin{figure}[H]
   \centering
   \includegraphics[width=0.75\columnwidth]{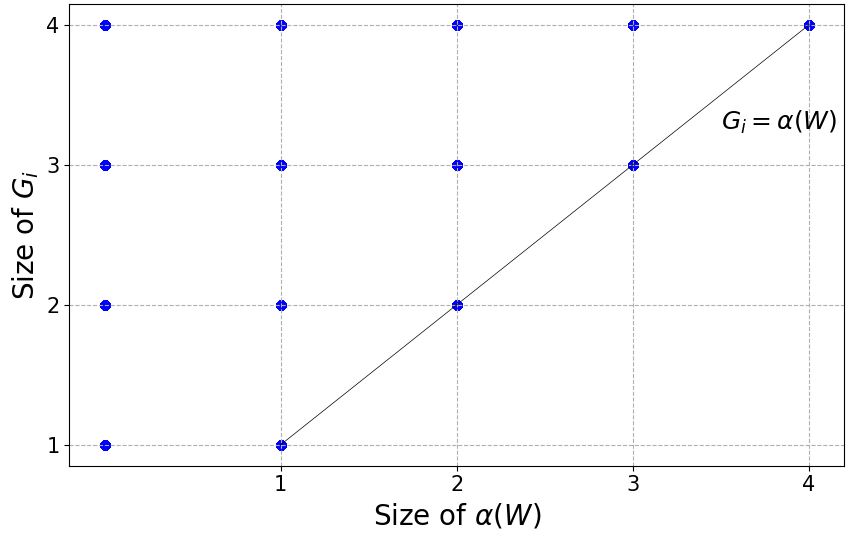}
   \caption{Number of agents in the feasible set $\alpha(W)$ compared to the true influencer set $G_i$ for networks with at most $5$ agents. The \texttt{Waterfall} is run using the tie-breaking subroutine from Algorithm \ref{alg: Waterfall algorithm} in Section \ref{sec: tie-breaking subroutine}.}
   \label{fig: size ambit multi filter}
\end{figure}

\section{Parameter Choice for Graph Generation Models}
\label{sec: param specs}
We test our algorithms on four random network structures: Erdős–Rényi (ER), Watts–Strogatz (WS), Regular Graph (RG), and Barabási–Albert (BA). All graphs are generated using the Python library NetworkX (version 3.5) with Python version 3.11. For fair comparison, given a sparsity regime $p\in \{0.1, 0.25, 0.5, 0.75, 0.9\}$, we define $\texttt{p$_1$} = p$ and $\texttt{p$_2$} = 2 + p\cdot(\frac{n}{2} - 1)$, rounding \texttt{p$_2$} to the nearest even integer. All graphs are generated using the functions and parameters detailed below and then converted to their directed versions (i.e., undirected edges become bidirectional). Omitted parameters use default values. 
\begin{itemize}
    \item[-] \textbf{ER:} \texttt{gnp\_random\_graph(n, p=p$_1$)},
    \item[-] \textbf{WS:} \texttt{watts\_strogatz\_graph(n, k=p$_2$, p=p$_1$)},
    \item[-] \textbf{RG:} \texttt{random\_regular\_graph(d=p$_2$, n)} and
    \item[-] \textbf{BA:} \texttt{barabasi\_albert\_graph(n, m=p$_2$)}.
\end{itemize}

For a network of size $n\geq 1$, labelling samples were created using NumPy (version 1.24) via \texttt{np.random.choice([-1, 1], size=n)}. Individual labellings were appended to a set to ensure no duplicates until the desired sample size $m$ was reached. Within the \texttt{Waterfall} algorithm, ties were resolved uniformly at random, again via \texttt{np.random.choice}. Parallel execution was handled using Python’s built-in \texttt{concurrent.futures}. 

To ensure reproducibility, a global \texttt{master\_seed} was set at the start of each experiment with parameters $(n,m)$. Each subprocess had its unique seed from offsetting the master seed with the process index (i.e., \texttt{seed = master\_seed + process\_id}), so that each parallel worker drew from an independent and deterministic stream. We applied this consistently to NetworkX’s graph generators, NumPy’s sampling, and the \texttt{Waterfall}'s tie-breaks. Experiments were run on Ubuntu 22.04 with 16GB RAM and a $12^{th}$ Gen Intel Core i7-1260P CPU (16 threads), using Visual Studio Code (version 1.102). No dedicated GPU was used. 

Experimental results, differentiated by network type and sparsity regime, are shown in Figure \ref{fig: Waterfall algorithm fails - network type} and Figure \ref{fig: Waterfall algorithm fails - sparsity regime}, respectively. 

\begin{figure}[ht]
    \centering
    \includegraphics[width=\columnwidth]{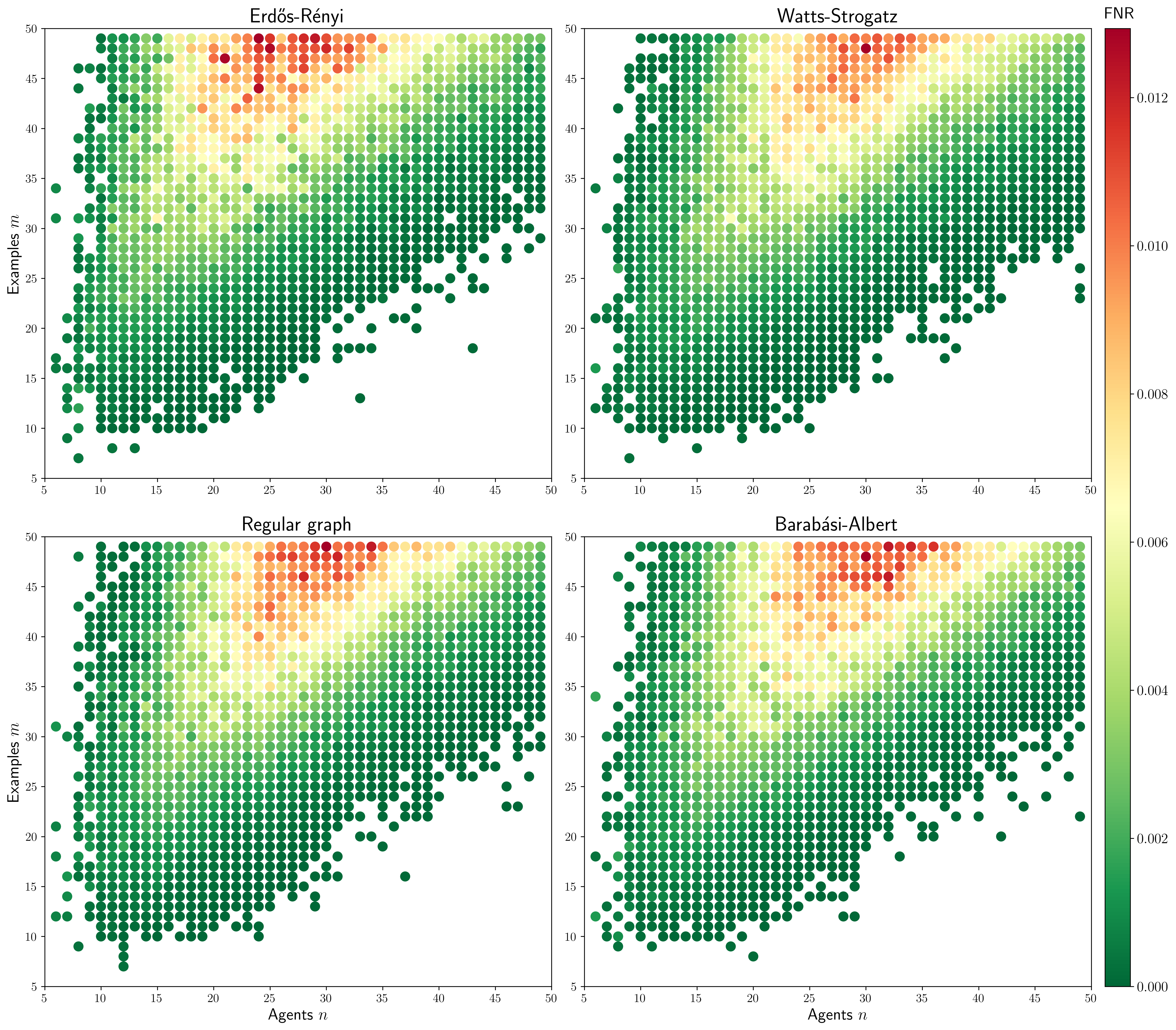}
    \caption{\texttt{Waterfall} FNR by network type. Each grid point corresponds to tests on 40 networks with~$n$ agents per sparsity value, on 50 shared labelling samples of size $m$.}
    \label{fig: Waterfall algorithm fails - network type}
\end{figure} 

\begin{figure}[H]
    \centering
    \includegraphics[width=\columnwidth]{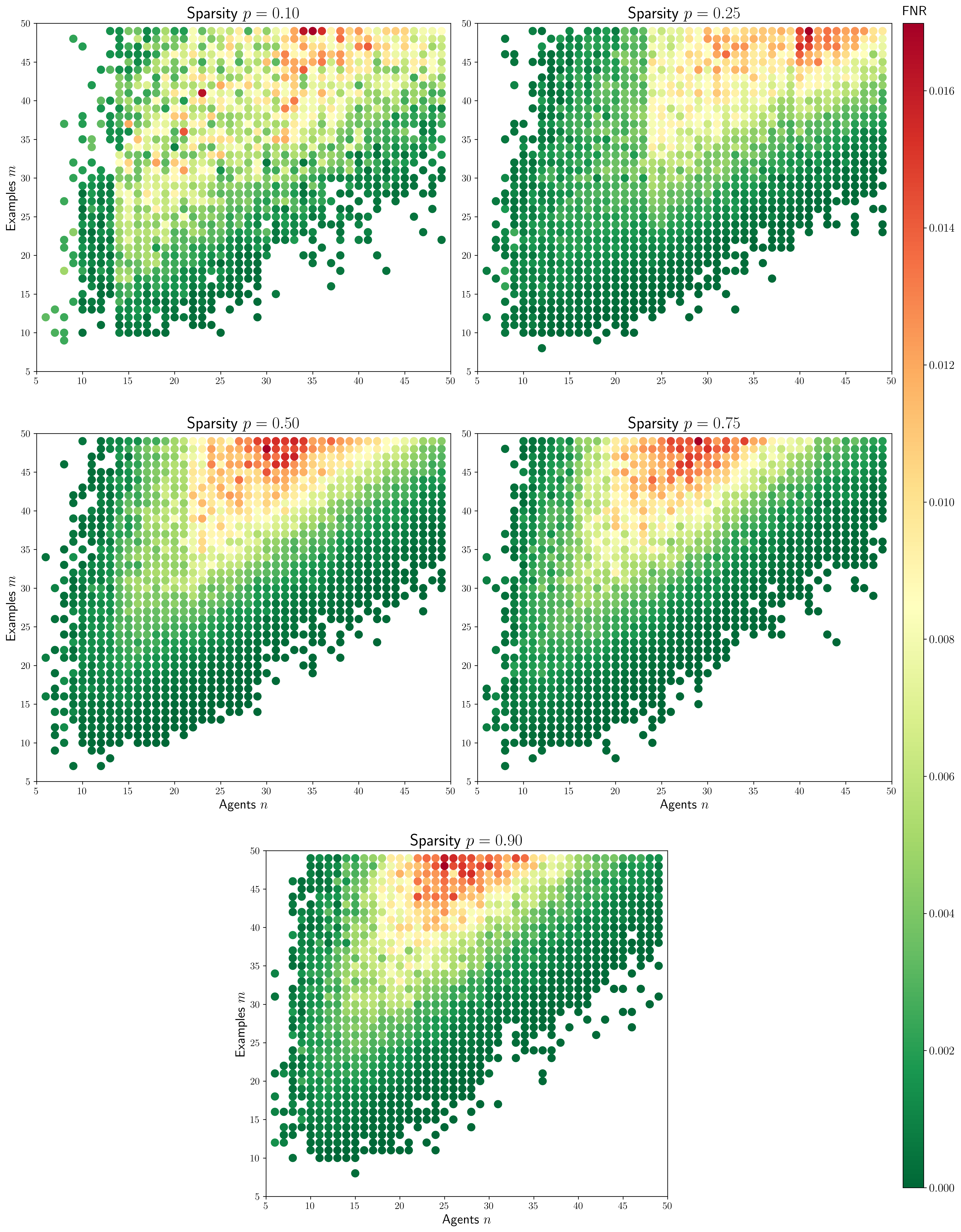}
    \caption{\texttt{Waterfall} FNR by sparsity regime. Each grid point corresponds to tests on 10 networks with~$n$ agents per graph type, on 50 shared labelling samples of size $m$.}
    \label{fig: Waterfall algorithm fails - sparsity regime}
\end{figure} 




\end{appendices}


\bibliography{references}

\end{document}